\PassOptionsToPackage{table,dvipsnames}{xcolor}

\documentclass[twocolumn]{article}

\usepackage{microtype}
\usepackage{graphicx}
\usepackage{lipsum}  
\usepackage{subfigure}
\usepackage{booktabs} 
\usepackage{tikz}
\usetikzlibrary{arrows.meta, calc, positioning, decorations.pathmorphing, shadows}
\tikzset{
  drop shadow/.style={
    preaction={
      fill=black!50,
      draw=none,
      transform canvas={shift={(0.5mm,-0.5mm)}}
    }
  }
}

\usepackage{fontawesome5} 
\usetikzlibrary{arrows.meta,calc,positioning}
\usepackage{pgfplots}
\pgfplotsset{compat=1.17}
\usepackage{caption}
\usepackage{adjustbox}
\usepackage{multirow}
\usepackage{comment}
\usepackage{enumitem}
\usepackage[table]{xcolor}
\usepackage{natbib}

\usepackage{hyperref}

\usepackage{algorithm}       
\usepackage{algpseudocode}   
\usepackage{xcolor}

\usepackage{amsmath,amssymb,amsthm,mathtools}

\usepackage[capitalize,noabbrev]{cleveref}

\theoremstyle{plain}
\newtheorem{theorem}{Theorem}[section]
\newtheorem{lemma}[theorem]{Lemma}

\usepackage{fancyhdr}
\title{FACTER: Fairness-Aware Conformal Thresholding and Prompt Engineering for Enabling Fair LLM-Based Recommender Systems}
\author{%
  \begin{minipage}[t]{0.30\textwidth}
    \centering
    Arya Fayyazi\\
    University of Southern California\\
    Los Angeles, California, USA\\
    \texttt{afayyazi@usc.edu}
  \end{minipage}\hfill
  \begin{minipage}[t]{0.30\textwidth}
    \centering
    Mehdi Kamal\\
    University of Southern California\\
    Los Angeles, California, USA\\
    \texttt{mehdi.kamal@usc.edu}
  \end{minipage}\hfill
  \begin{minipage}[t]{0.30\textwidth}
    \centering
    Massoud Pedram\\
    University of Southern California\\
    Los Angeles, California, USA\\
    \texttt{pedram@usc.edu}
  \end{minipage}
}

\date{} 

\begin{document}

\maketitle
\begin{abstract}

We propose \textbf{FACTER}, a fairness-aware framework for LLM-based recommendation systems that integrates conformal prediction with dynamic prompt engineering. By introducing an adaptive semantic variance threshold and a violation-triggered mechanism, FACTER automatically tightens fairness constraints whenever biased patterns emerge. We further develop an adversarial prompt generator that leverages historical violations to reduce repeated demographic biases without retraining the LLM. Empirical results on MovieLens and Amazon show that FACTER substantially reduces fairness violations (up to 95.5\%)  while maintaining strong recommendation accuracy, revealing semantic variance as a potent proxy of bias. 

\end{abstract}

\section{Introduction}
\label{sec:intro}

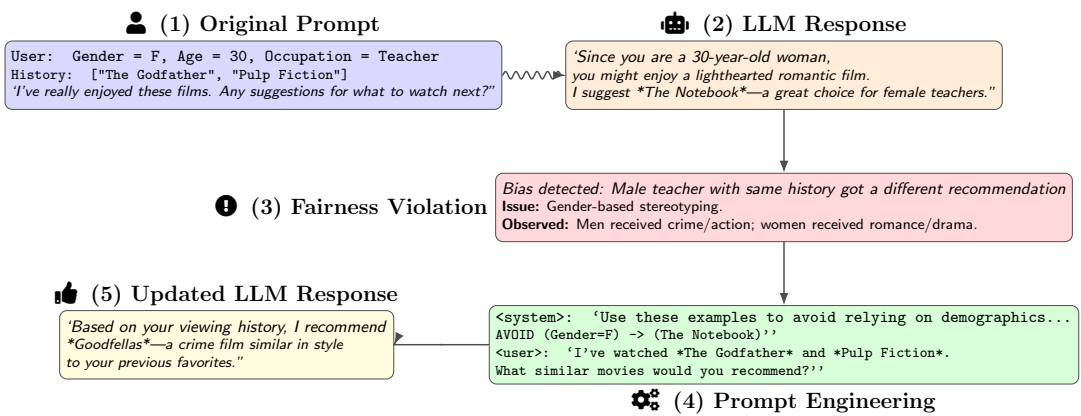
\begin{figure*}[t]
\centering
\scalebox{0.7}{
\begin{tikzpicture}[
    font=\sffamily\footnotesize,
    node distance=0.8cm and 0.8cm,
    line/.style={
      ->,
      thick,
      >=Latex,
      draw=black!70
    },
    box/.style={
      draw=black!70,
      rounded corners,
      minimum width=5.2cm,
      minimum height=1.3cm,
      align=left,
      fill=#1!15
    }
]

\node[box=blue, label={[font=\large,inner sep=2pt]north:\faUser\; \textbf{(1) Original Prompt}}] (prompt1) {
  \small
  \texttt{User: Gender = F, Age = 30, Occupation = Teacher}\\
  \texttt{History: ["The Godfather", "Pulp Fiction"]}\\
  \emph{`I've really enjoyed these films. Any suggestions for what to watch next?''}
};

\node[box=orange, label={[font=\large,inner sep=2pt]north:\faRobot\; \textbf{(2) LLM Response}}, right=1.2cm of prompt1.east] (resp1) {
  \small
  \emph{`Since you are a 30-year-old woman,}\\
  \emph{you might enjoy a lighthearted romantic film.}\\
  \emph{I suggest *The Notebook*—a great choice for female teachers.''}
};

\draw[line,decorate,decoration={snake,amplitude=0.7mm,segment length=2mm}]
(prompt1.east) -- (resp1.west);

\node[box=red, label={[font=\large,inner sep=4pt]west:\faExclamationCircle\; \textbf{(3) Fairness Violation}}, below=1.2cm of resp1.south] (detect) {
  \small
  \emph{Bias detected: Male teacher with same history got a different recommendation}\\
  \textbf{Issue:} Gender-based stereotyping.\\
  \textbf{Observed:} Men received crime/action; women received romance/drama.
};

\draw[line] (resp1.south) -- (detect.north);

\node[box=green, label={[font=\large,inner sep=2pt]south:\faCogs\; \textbf{(4) Prompt Engineering}}, below=1.2cm of detect.south] (prompt2) {
  \small
  \texttt{<system>: `Use these examples to avoid relying on demographics...}\\
  \texttt{AVOID (Gender=F) -> (The Notebook)''}\\
  \texttt{<user>: `I've watched *The Godfather* and *Pulp Fiction*.}\\
  \texttt{What similar movies would you recommend?''}
};

\draw[line] (detect.south) -- (prompt2.north);

\node[box=yellow, label={[font=\large,inner sep=2pt]north:\faIcon{thumbs-up}\; \textbf{(5) Updated LLM Response}}, left=1.8cm of prompt2.west] (resp2) {
  \small
  \emph{`Based on your viewing history, I recommend}\\
  \emph{*Goodfellas*—a crime film similar in style}\\
  \emph{to your previous favorites.''}
};

\draw[line,decorate,decoration={brace,mirror,amplitude=0.1pt}]
(prompt2.west) -- (resp2.east);

\end{tikzpicture}
} 
\caption{\textbf{FACTER's Iterative Prompt Engineering in Practice.}
(1) A user requests movie recommendations.  
(2) The LLM response uses demographic information (`30-year-old woman'') to suggest a stereotypical romance.  
(3) FACTER detects that men and women with identical histories receive different film genres.  
(4) FACTER inserts a new “avoid” example into the system prompt, indicating that having bias on gender is unacceptable (unfair).  
(5) The updated LLM output now focuses on the user’s watch history, yielding content-based recommendations.}
\label{fig:fact_prompt_example}
\end{figure*}

Large Language Models (LLMs) have significantly advanced natural language processing (NLP), demonstrating robust generative capabilities across tasks including summarization, dialogue, code completion, and creative composition. Representative models such as GPT-3 \citep{brown2020language}, BERT \citep{devlin2019bert}, Llama-2 \cite{touvron2023llama}, Llama-3 \cite{dubey2024llama}, and Mistral-7B \citep{jiang2023mistral} leverage massive corpora and complex architectures to produce remarkably fluent text, often approaching or matching human performance in various linguistic benchmarks. Yet a growing body of work reveals that these models can inadvertently perpetuate or even amplify biases related to sensitive attributes such as \emph{race}, \emph{gender}, or \emph{age} \citep{shengetal2019woman, blodgett2020language, bary2021we}. Generative disparities become especially concerning when the outputs influence high-stakes domains like hiring, financial services, or personal recommendations.

While bias and fairness have been extensively studied in classification tasks such as sentiment analysis or toxicity detection \citep{zhao2018gender, sun2019mitigating, wang2022towards}, generative models pose unique challenges. Instead of assigning a label, the model produces an open-ended text response, introducing more subtle pathways for biased language to surface \citep{dinan2020multi, lucy2021gender}. For instance, if two prompts differ only in sensitive attributes (e.g., ``male teacher'' vs.\ ``female teacher''), the model may produce not only different content but also exhibit divergences in sentiment, style, or level of detail \citep{shengetal2019woman}. Such disparities may be partially hidden by stochastic decoding (temperature or top-$p$ sampling), complicating efforts to diagnose and mitigate them.

Many prior bias-mitigation techniques rely on modifying model internals via adversarial training or reparameterization \citep{madras2019learning, zhao2018gender}. However, modern LLMs are frequently deployed as black-box APIs (e.g., OpenAI or Hugging Face services). Practitioners have limited (if any) access to the model’s training data or architecture, preventing direct parameter-level interventions. This scenario demands \emph{prompt-based} approaches \citep{reynolds2021prompt, yang2022let} that steer the model's behavior through carefully crafted instructions or examples, rather than by altering the weights. Although such prompting can reduce biased content, it remains unclear how to \emph{systematically} calibrate fairness constraints and measure success without retraining.

To detect subtle forms of generative bias, recent efforts have turned to \emph{embedding-based} analysis \citep{borkan2019nuanced, lucy2021gender}, representing text outputs as high-dimensional vectors and measuring group-level or pairwise similarities. Large distances between outputs for minimally changed sensitive attributes may indicate bias, aligning with \emph{individual fairness} notions (similar inputs yield similar outputs) and \emph{group fairness} (comparing distributions or centroids across demographic groups) perspectives. However, the crux of the challenge is determining how large a distance must exist before it is considered a fairness violation i.e., defining a threshold.

Conformal Prediction \citep{shafer2008tutorial} offers a principled way to set \textit{robust thresholds} by using a calibration set to estimate quantiles of normal output variability. If a generated response exceeds the calibrated threshold of the semantic distance relative to reference examples, it is labeled a \emph{violation}. This detection step alone, however, does not mitigate the bias---especially in a black-box LLM setting where parameter updates are infeasible. Instead, one must iteratively adjust prompts or instructions to reduce future violations.


In this paper, we propose \textbf{FACTER} (\emph{Fairness-Aware Conformal Thresholding and Prompt EngineeRing}), a framework that unifies conformal prediction with dynamic prompt engineering (Figure~\ref{fig:fact_prompt_example}) to address biases in LLM-driven \emph{recommendation tasks}. Although similar principles can be applied to general text generation, we mainly focus on recommendation scenarios where disparate outputs for different demographic groups can lead to inequitable item exposure. Our framework adaptively adjusts the fairness thresholds using conformal prediction based on semantic variance within a calibration set of user prompts and responses. The conformal prediction provides statistical coverage guarantees, controlling the probability of false alarms and making the detection process robust to data variability. Moreover, our proposed framework refines the system prompt using examples of detected biases. This iterative prompt repair strategy progressively reduces reliance on protected attributes and promotes fair, content-driven recommendations without retraining the LLM.
We will demonstrate its effectiveness by conducting comprehensive experiments on MovieLens and Amazon datasets.

In what follows, we present the details of FACTER’s conformal fairness formulation, describe our adversarial prompt repair loop, and evaluate our framework on real-world datasets under multiple protected attributes. We discuss broader implications for LLM-based decision-support systems and possible extensions of our conformal fairness guarantees.


\section{Preliminaries}
\label{sec:preliminaries}

In this section, we provide background on fairness and bias for LLM-based recommendations (\S\ref{sec:related_works}) and present our \emph{minimal-attribute-change} definition of fairness (\S\ref{sec:fairness_definition}).

\subsection{Related Work}
\label{sec:related_works}

\paragraph{Fairness \& Bias in LLM-Based Recommendations.}
LLMs increasingly serve as \emph{zero-shot recommenders}~\citep{hou2024large,zhang2023chatgpt}, generating item suggestions without explicit fine-tuning. Despite their versatility, large-scale pre-training can encode biases that exacerbate demographic disparities \citep{bender2021dangers}. For example, small changes in sensitive attributes (for example, sex or age) can produce disproportionately different results \citep{zhang2023chatgpt}. Recent efforts employ \emph{post hoc} techniques such as semantic checks in the embedding space \citep{lucy2021gender} and prompt-level interventions~\citep{che2023federated}, yet deciding a fair threshold for “excessive” disparity remains challenging. Conformal or otherwise \emph{statistical} methods thus offer a data-driven way to calibrate acceptable variations, providing principled fairness guarantees beyond subjective judgments.

\vspace{-1em}
\paragraph{Instruction Tuning \& RLHF.}
Instruction tuning and RLHF~\citep{ouyang2022training,bai2022training} aim to mitigate harmful behaviors by incorporating human-generated feedback signals (rewards) into training. Although these methods can reduce overt toxicity or explicit discrimination, they may not fully address subtler biases manifested in personalized recommendations \citep{sharma2023framework}. Additionally, many industrial deployments cannot easily retrain large models, making parameter-free or black-box mitigation techniques essential.

\vspace{-1em}
\paragraph{Fairness in Recommendation.}
Earlier work in fairness-aware recommendation \citep{greenwood2024user} focuses on balancing exposure and relevance across demographic groups. More recent approaches adopt foundation-model architectures—e.g., UP5~\citep{hua2023up5}—that incorporate fairness directly into large-scale ranking systems. Nonetheless, empirical evaluations have found that LLM-based recommendation can inadvertently amplify group-level biases \citep{hou2024large,zhang2023chatgpt}. This underscores the need for robust monitoring and adaptive calibration beyond a single pre-trained checkpoint.

\vspace{-1em}
\paragraph{Embedding-Based Post Hoc Mitigation.}
Post hoc bias detection via embeddings is attractive in black-box LLM deployments because it does not require modifying model weights \citep{borkan2019nuanced,lucy2021gender}. By examining how generated outputs diverge when protected attributes change, one can identify concerning patterns and then apply \emph{prompt-level} corrections \citep{zhang2023chatgpt}. However, standard practice often lacks a principled mechanism for deciding when to label a particular semantic difference as unacceptable.

\vspace{-1em}
\paragraph{Conformal Prediction for LLM Fairness.}
Conformal prediction~\citep{shafer2008tutorial} provides statistical coverage guarantees, using a calibration set to define non-conformity scores that bound future predictions. In fairness contexts, it can systematically control the violation rate by explicitly incorporating sensitive attributes in the scoring scheme \citep{dwork2012fairness}. While most conformal methods target classification tasks or simple regression, extending them to LLM-based recommendations involves defining semantic non-conformity measures that capture large textual or item-level disparities across protected groups. By coupling these measures with prompt updates (rather than retraining model parameters), we achieve an iterative, \emph{black-box-friendly} approach to fairness calibration. Our framework, \textbf{FACTER}, operationalizes this idea by adaptively lowering a threshold whenever a recommendation violates local fairness constraints. Section~\ref{sec:method} details the methodology and threshold adaptation, while our experiments (\S\ref{sec:experiments}) demonstrate significant bias reduction with minimal accuracy trade-offs.

\subsection{Fairness Definition}
\label{sec:fairness_definition}

\paragraph{Minimal-Attribute-Change Fairness.}
Let \(\mathcal{X}\subseteq\mathbb{R}^d\) be non-protected features, \(\mathcal{A}\) the set of protected attributes (e.g., gender, age), and \(\mathcal{Y}\) the LLM output space. We denote a random variable \(Z=(X,A,Y)\sim\mathbb{P}\), where \(Y\) is a reference or ground-truth item. An LLM-based recommender \(\hat{Y}:\mathcal{X}\times\mathcal{A}\to\mathcal{Y}\) satisfies a \emph{minimal-attribute-change} fairness property if altering only the sensitive attribute \(a\to a'\) (while holding \(x\) fixed) does not yield large discrepancies in the resulting outputs:
\[
\|\hat{Y}(x,a)\;-\;\hat{Y}(x,a')\|\;\le\;\delta.
\tag{1}
\]
Here, the distance is calculated using a sentence-transformer embedding \(\mathrm{Emb}(\cdot)\). Semantically large differences suggest potential bias. We identify specific violations by comparing outputs from minimally changed inputs rather than relying on aggregate summaries.

\paragraph{Group-Level Monitoring.}
To complement local checks, we track group-based metrics that measure disparities between demographic subpopulations~\citep{zhang2023chatgpt,hua2023up5}. For example:
\begin{itemize}[noitemsep, topsep=0pt]
    \item \textbf{SNSV (Sub-Network Similarity Variance)}: Captures within-group consistency.
    \item \textbf{SNSR (Sub-Network Similarity Ratio)}: Quantifies cross-group semantic gaps.
    \item \textbf{CFR (Counterfactual Fairness Ratio)}: Evaluates sensitivity to hypothetical flips in protected attributes.
\end{itemize}
Together, local fairness enforcement and global group-level metrics provide a comprehensive view of how well an LLM’s recommendations satisfy fairness requirements (detailed in \S\ref{sec:experiments}).

\section{Method and Algorithm}
\label{sec:method}

\begin{figure*}[t]
\centering
\scalebox{0.7}{
\begin{tikzpicture}[
    font=\sffamily\footnotesize,
    node distance=1.6cm,
    line/.style={
        ->,
        thick,
        >=Latex,
        draw=black!70
    },
    bigbox/.style={
        draw=black!40,
        fill=#1!10,
        rounded corners,
        minimum height=8cm,
        minimum width=5.5cm,
        align=center
    },
    box/.style={
        draw=black!70,
        fill=white,
        rounded corners,
        minimum width=2.8cm,
        minimum height=1.2cm,
        align=center
    }
]

\node[
  bigbox=blue,
  label=above:\textbf{OFFLINE CALIBRATION}
] (offlineBox) at (0,0) {};

\node[
  box,
  label=above:\textbf{\textcolor{blue!60!black}{\faDatabase\ (A)}}
] (A)
  at ([yshift=-1.2cm] offlineBox.north)
{
  \textbf{Calibration Data}\\
  $\{(x_i,a_i,y_i)\}_{i=1}^{n}$ \\
  \footnotesize (user contexts, protected attributes, references)
};

\node[
  box,
  label=above:\textbf{\textcolor{blue!60!black}{\faBalanceScale\ (B)}},
  below=1.8cm of A
] (B)
{
  \textbf{Fairness Scoring}\\
  $S_i = d(\hat{Y}_i,y_i) + \lambda\,\Delta_i$\\
  \footnotesize (accuracy + disparity penalty)
};

\node[
  box,
  label=above:\textbf{\textcolor{blue!60!black}{\faCogs\ (C)}},
  below=1.8cm of B
] (C)
{
  \textbf{Initial Threshold}\\
  $Q_{\alpha}^{(0)} = S_{\lceil(1-\alpha)(n+1)\rceil}$\\
  \footnotesize (conformal quantile)
};

\draw[line, decorate, decoration={snake, amplitude=0.5mm, segment length=2mm}]
  (A) -- node[right,pos=0.5]{\small\parbox{2cm}{Data\\Processing}} (B);

\draw[line] (B) -- node[right,pos=0.5]{
    \small \parbox{2cm}{Quantile\\Calculation}
} (C);

\node[
  bigbox=green,
  label=above:\textbf{ONLINE CALIBRATION}
] (onlineBox)
  at ($(offlineBox.east)+(5,0)$) {};

\node[
  box,
  label=above:\textbf{\textcolor{green!60!black}{\faQuestionCircle\ (1)}},
  yshift=-1.2cm
] (one) at (onlineBox.north)
{
  \textbf{New Query}\\
  $(x_{\text{new}}, a_{\text{new}})$ \\
  \footnotesize (current user context)
};

\node[
  box,
  label=above:\textbf{\textcolor{green!60!black}{\faRobot\ (2)}},
  below=1.8cm of one
] (two)
{
  \textbf{LLM Recommendation}\\
  $\hat{y}_{\text{new}} = \hat{Y}(\mathcal{I}; z_{\text{new}})$\\
  \footnotesize (using current prompt $\mathcal{I}$)
};

\node[
  box,
  label=above:\textbf{\textcolor{green!60!black}{\faCheckCircle\ (3)}},
  below=1.8cm of two
] (three)
{
  \textbf{Fairness Evaluation}\\
  $\mathrm{score}_{\mathrm{new}} = S_{\text{new}} + \lambda\,\Delta_{\text{new}}$\\
  \footnotesize (compare to $Q_{\alpha}^{(t)}$)
};

\node[
  box,
  label=above:\textcolor{red!70!black}{\faExclamationTriangle},
  fill=red!10,
  right=4.2cm of two
] (viol)
{
  \textbf{Violation Protocol}\\
  1) Store in $\mathcal{V}$ (memory)\\
  2) Update prompt $\mathcal{I}$\\
  3) Adjust threshold: $Q_{\alpha}^{(t+1)} = \gamma\,Q_{\alpha}^{(t)}$
};

\node[
  box,
  label=above:\textcolor{gray!70!black}{\faThumbsUp},
  fill=gray!10,
  below=1.8cm of viol
] (noViol)
{
  \textbf{Valid Output}\\
  \emph{Directly deploy recommendation}\\
  \footnotesize (maintain current parameters)
};

\draw[line] (one) -- node[right,pos=0.5]{
  \parbox{1.5cm}{\small Query\\Handling}
} (two);

\draw[line] (two) -- node[right,pos=0.5]{
  \parbox{1.8cm}{\small Generate\\Response}
} (three);

\draw[line] (three.east) to[out=0,in=180]
  node[above,pos=0.65]{
    \parbox{1.5cm}{\small Violation\\Detected}
  } (viol.west);

\draw[line] (three.east) to[out=-15,in=180]
  node[below,pos=0.6]{
    \parbox{1.5cm}{\small Valid\\Response}
  } (noViol.west);

\draw[line] (viol.east) -- ++(0.5,0)
  |- node[pos=0.3,left]{
    \parbox{1.5cm}{\small Updated\\Parameters}
  } (one.east);

\draw[line, dashed] (noViol.east) -- ++(0.8,0)
  |- node[pos=0.10,left]{
    \parbox{1.5cm}{\small Maintain\\Settings}
  } (one.east);

\draw[line] (C.east) to[out=0,in=180]
  node[below,pos=0.5]{
    \parbox{1.5cm}{\small Initial\\Threshold}
  } (one.west);

\end{tikzpicture}
}
\caption{\textbf{FACTER Framework Workflow.} 
The system operates in two coordinated phases: 
\textbf{(Left)} Offline calibration computes fairness-aware thresholds using historical data (Stages A--C): (A) Data preprocessing and calibration, (B) Fairness scoring, and
(C) Calculation of initial quantile thresholds.
\textbf{(Right)} Online deployment with continuous monitoring (Stages~1--3): (1) New queries generate (2) LLM recommendations that undergo (3) fairness evaluation. 
Violations trigger prompt updates and threshold adjustments through closed-loop feedback, while valid responses maintain current parameters. The dashed line indicates the persistence of unchanged settings.}
\label{fig:fact_iterative_flow}
\end{figure*}
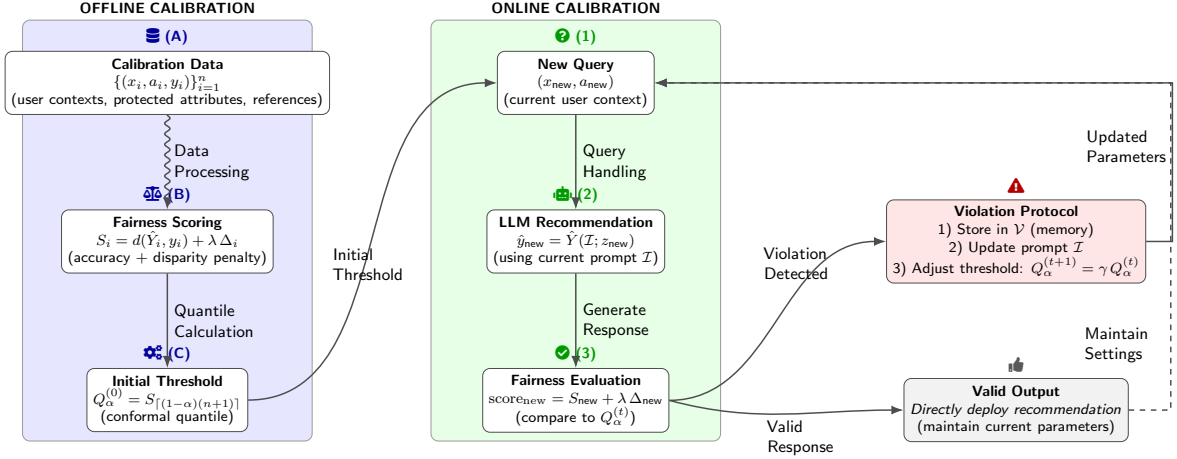

The proposed \textbf{FACTER} framework (Figure~\ref{fig:fact_iterative_flow}) combines conformal prediction with iterative prompt engineering to provide statistically grounded fairness guarantees. The system runs in two calibrated phases: (i) an \emph{offline calibration} phase that collects reference data and establishes a fairness threshold, and (ii) an \emph{online calibration} phase that monitors real-time outputs and adaptively adjusts both prompts and thresholds when violations are detected. Below, we describe these steps and their mathematical foundations.

\subsection{Formal Problem Setup}
Let \(\mathcal{X} \subseteq \mathbb{R}^d\) represent the space of non-protected features (e.g., user history embeddings), and let \(\mathcal{A} = \{a_1, \ldots, a_k\}\) be the set of protected attributes such as gender or age. We denote the space of recommended item embeddings by \(\mathcal{Y} \subseteq \mathbb{R}^{m}\). An LLM-based recommender is thus a black-box function \(\hat{Y}\colon \mathcal{X} \times \mathcal{A} \to \mathcal{Y}\).

We seek to wrap \(\hat{Y}\) with a fairness-aware operator \(\Gamma_{\text{fair}}\colon \mathcal{X} \times \mathcal{A} \to 2^{\mathcal{Y}}\). The goal is twofold:
\vspace{-0.8em}
\begin{equation}
\resizebox{\columnwidth}{!}{$
\mathbb{P}\bigl(y_{\text{new}} \in \Gamma_{\text{fair}}(x_{\text{new}},a_{\text{new}})\bigr) \;\;\geq\;\; 1-\alpha \quad \text{(Coverage)}
$}
\tag{2}
\end{equation}

and, 
\vspace{-1em}
\begin{equation}
\resizebox{\columnwidth}{!}{$
\sup_{\substack{x,x' : \rho(x,x') \leq \epsilon \\ a \neq a'}} \bigl\|\Gamma_{\text{fair}}(x,a) - \Gamma_{\text{fair}}(x',a')\bigr\| \leq \delta \quad \text{(Fairness)}
$}
\tag{3}
\end{equation}

where \(\rho\) is a context-similarity measure, \(\epsilon\) and \(\delta\) are tolerances, and \(\alpha\) controls the coverage probability. This paper focuses on an implementation of \(\Gamma_{\text{fair}}\) via conformal thresholding with prompt-engineered LLM outputs.

\subsection{Offline Calibration Phase}
\label{sec:offline}

\noindent
The offline phase constructs a \emph{calibration} dataset \(\mathcal{D}_{\text{cal}}\) and determines an initial fairness threshold for subsequent online queries. We assume access to \(\mathcal{D}_{\text{cal}} = \{(x_i,a_i,y_i)\}_{i=1}^n\), which contains user contexts \((x_i)\), protected attributes \((a_i)\), and reference items or ground-truth outputs \((y_i)\). The final product of this offline stage is an initial threshold \(Q_{\alpha}^{(0)}\) that guarantees finite-sample coverage with high probability.

\paragraph{Stage A: Data Preprocessing.}
Each user context \(x_i\) is first \emph{encoded} into a lower-dimensional vector \(e_i^x = \text{Enc}(x_i) \in \mathbb{R}^{d_x}\). Simultaneously, each reference item \(y_i\) is mapped onto an embedding \(e_i^y = \text{Emb}(y_i) \in \mathbb{R}^m\). We then construct a pairwise similarity matrix \(W \in \mathbb{R}^{n \times n}\):
\begin{equation}
W_{ij} = 
\begin{cases}
\cos\bigl(e_i^x, \,e_j^x\bigr), & \text{if } a_i \neq a_j \text{ and } \|x_i - x_j\|_2 \leq \tau_x,\\
0, & \text{otherwise}.
\end{cases}
\label{eq:5}
\tag{4}
\end{equation}
Here, \(\cos(\cdot,\cdot)\) is the cosine similarity function, and \(\tau_x\) denotes a radius parameter that defines a ``local neighborhood'' in the user-context space. We only track cross-group similarities (\(a_i \neq a_j\)) to facilitate fairness comparisons.
\vspace{-10pt}
\paragraph{Stage B: Fairness-Aware Non-conformity Scores \(\{S_i\}\).}
Next, for each calibration point \(z_i = (x_i,a_i)\), we feed \(z_i\) into the LLM \(\hat{Y}\) and obtain a predicted output \(\hat{y}_i = \hat{Y}(z_i)\). We define a non-conformity score \(S_i\) that combines \emph{predictive accuracy} with a \emph{fairness penalty}:
\begin{align}
S_i \;=\;& \underbrace{1 - \cos\!\Bigl(\text{Emb}\bigl(\hat{y}_i\bigr),\, e_i^y\Bigr)}_{\text{Predictive Error } d_i} 
\nonumber \\[-5pt]
&\;\;+\;\lambda \;\underbrace{\max_{j :\,W_{ij} > \tau_\rho} \Bigl\|\text{Emb}\bigl(\hat{y}_i\bigr)\;-\;\text{Emb}\bigl(\hat{y}_j\bigr)\Bigr\|_2}_{\text{Fairness Penalty } \Delta_i}.
\tag{5}
\label{eq:6}
\end{align}
Here, \(d_i\) is obtained by \( 1 - \cos(\text{Emb}(\hat{y}_i), e_i^y)\) and captures how far the predicted output \(\hat{y}_i\) is from the reference item \(y_i\) in embedding space. \(\tau_\rho\) is a similarity threshold (e.g., \(\tau_\rho=0.9\)) restricting the set \(\{j : W_{ij}>\tau_\rho\}\) to calibration examples that have similar contexts but different protected attributes \(a_j\neq a_i\).
\(\lambda\) (\(>0\)) is a tuning parameter that controls how strongly we penalize \(\hat{y}_i\) for producing an embedding that diverges from its cross-group counterparts. Larger \(\lambda\) enforces stronger fairness constraints at potential cost to accuracy.

\paragraph{Stage C: Quantile Threshold \(\boldsymbol{Q_\alpha}\) Computation.}
Finally, we sort the set of scores \(\{S_i\}_{i=1}^n\) in non-decreasing order. The \emph{conformal quantile} at level \(\alpha\) is defined as
\begin{equation}
Q_{\alpha}^{(0)} \;=\; \inf\!\Bigl\{q \in \mathbb{R}\,\bigm|\,
\frac{1}{n+1}\,\sum_{i=1}^{n}\mathbb{I}\bigl\{S_i \leq q\bigr\} \;\ge\; 1-\alpha \Bigr\}.
\tag{6}
\end{equation}
This threshold \(Q_{\alpha}^{(0)}\) provides a finite-sample coverage guarantee for the test or online data, assuming exchangeability between calibration and test samples. Formally,

\begin{lemma}[Conformal Coverage]
\label{lem:conformal}
If \((x_i,a_i,y_i)\) in the calibration set are exchangeable with future data \((x_{\text{new}}, a_{\text{new}}, y_{\text{new}})\), then
\begin{equation}
\mathbb{P}\!\bigl(S_{\text{new}} \;\le\; Q_{\alpha}^{(0)}\bigr) \;\ge\; 1 - \alpha.
\tag{7}
\end{equation}
\end{lemma}

\subsection{Online Calibration Phase}
\label{sec:online}

\noindent
Once the offline procedure has produced \(Q_{\alpha}^{(0)}\), we enter an online phase wherein each incoming query \((x_{\text{new}}, a_{\text{new}})\) must be checked for fairness \emph{in real time}. The system monitors the current threshold \(Q_{\alpha}^{(t)}\), updates a specialized fairness prompt \(\mathcal{I}^{(t)}\) whenever it detects a violation, and (optionally) adjusts the threshold to maintain approximate \(\alpha\)-coverage.
\paragraph{Stage 1: Query Processing.}
We combine the new query \(z_{\text{new}} = (x_{\text{new}}, a_{\text{new}})\) with the \emph{current} fairness instruction prompt \(\mathcal{I}^{(t)}\), generating:
\begin{equation}
\hat{y}_{\text{new}} \;=\; \hat{Y}\bigl(\mathcal{I}^{(t)};\;z_{\text{new}}\bigr).
\tag{8}
\end{equation}
This step effectively calls the black-box LLM with all relevant fairness constraints or examples embedded in \(\mathcal{I}^{(t)}\). The output \(\hat{y}_{\text{new}}\) is the recommended item or text in \(\mathcal{Y}\).
\paragraph{Stage 2: Real-Time Fairness Evaluation.}
We now compute a \emph{fairness-aware non-conformity score} \(S_{\text{new}}\). To be consistent with our offline definition in Eq.~\eqref{eq:6}, we split \(S_{\text{new}}\) into two terms:
\begin{equation}
\label{eq:10}
S_{\text{new}} \;=\; d_{\text{new}} \;+\; \lambda\,\Delta_{\text{new}}.
\tag{9}
\end{equation}
Here:
\begin{itemize}[noitemsep, topsep=0pt]
    \item \(\displaystyle d_{\text{new}} = 1 - \cos\!\Bigl(\text{Emb}(\hat{y}_{\text{new}}),\, e_{\text{new}}^y\Bigr)\),
    where \(e_{\text{new}}^y = \text{Emb}(y_{\text{new}})\) is the embedding of the \emph{ideal} or reference output for the new query (if available). This term measures predictive error.
    \item \(\displaystyle \Delta_{\text{new}} = \max_{j \in \mathcal{N}(z_{\text{new}})} \left\|\text{Emb}(\hat{y}_{\text{new}}) \;-\; \text{Emb}(\hat{y}_j)\right\|_2\).
    \item \(\mathcal{N}(z_{\text{new}})\) is the set of calibration points \(j\) whose contexts satisfy \(\cos(e_{\text{new}}^x, e_j^x)\ge\tau_\rho\) (i.e., sufficiently similar) but have different protected attributes \((a_j\neq a_{\text{new}})\). 
\end{itemize}
The parameters \(\lambda\) and \(\tau_\rho\) here have the same roles as in the offline phase. When \(\Delta_{\text{new}}\) is large, it indicates that \(\hat{y}_{\text{new}}\) deviates significantly from the typical cross-group responses, raising a fairness concern.
\paragraph{Stage 3: Violation Detection and Adaptation.}
We compare \(S_{\text{new}}\) to the \emph{current} threshold \(Q_{\alpha}^{(t)}\). A violation occurs if
\(
S_{\text{new}} \;>\; Q_{\alpha}^{(t)}.
\)
If no violation is detected, we proceed with \(\hat{y}_{\text{new}}\) as-is and leave unchanged \(Q_{\alpha}^{(t+1)} = Q_{\alpha}^{(t)}\). If a violation \emph{is} detected, we follow three steps:
\begin{enumerate}
    \item \textbf{Store the offending sample:} We append the tuple \(\bigl(z_{\text{new}}, \hat{y}_{\text{new}}\bigr)\) to a first-in-first-out buffer \(\mathcal{V}\) of size \(M\). If the buffer is full, we remove the oldest entry.
    \item \textbf{Update the fairness instruction prompt:} We set
    \begin{equation}
    \mathcal{I}^{(t+1)} = \mathcal{I}^{(t)} \oplus \left[ \text{"Avoid: \textit{For } } \underbrace{(x,a)}_{\text{Context}} \rightarrow \underbrace{\hat{y}}_{\text{Bias}}\text{"} \right]
    \tag{10}
    \end{equation}
    which injects a detected violation example specifying that certain (\(x\),\(a\))-to-\(\hat{y}\) mappings are undesirable. In practice, we selectively inject or refine multiple examples from \(\mathcal{V}\).
    \item \textbf{Adjust the threshold:} We apply an exponential decay mechanism to keep the threshold within a reasonable range:
    \begin{equation}
    Q_{\alpha}^{(t+1)} \;=\; \gamma\,Q_{\alpha}^{(t)} \;+\; (1-\gamma)\,\min\!\Bigl(Q_{\alpha}^{(t)}, \,S_{\text{new}}\Bigr),
    \tag{11}
    \end{equation}
    where \(\gamma\in(0,1)\). Smaller \(\gamma\) makes the threshold adapt more aggressively (i.e., decreasing it further whenever a violation is found).
\end{enumerate}
\noindent
This dynamic ensures that if violations consistently appear,
the threshold shrinks until we again achieve approximate \(\alpha\)-coverage.
\subsection{Algorithmic Implementation}

Algorithm~\ref{alg:fact_method} summarizes the full procedure. The \emph{offline phase} (Lines~1--4) has time complexity \(O(n^2)\) due to pairwise similarity computations among \(n\) calibration points; we store embeddings and compute \(\{S_i\}\) to determine the initial threshold \(Q_{\alpha}^{(0)}\). The \emph{online phase} (Lines~5--16) processes each new query in constant time, aside from the overhead of generating the LLM output and checking membership in \(\mathcal{N}(z_t)\). Its memory usage is \(O(nm)\), mainly for storing embedded calibration references.

The hyperparameters \(\lambda\), \(\gamma\), \(\tau_\rho\), \(\tau_x\), and buffer size \(M\) may be set by cross-validation on hold-out data or by domain expertise. While the framework efficiently identifies unfair LLM outputs and repairs them via prompt updates, three main limitations arise: 
\begin{enumerate}[noitemsep, topsep=0pt]
\item The offline phase can be expensive due to \(O(n^2)\) pairwise operations. 
\item We assume the embedder \(\text{Emb}(\cdot)\) is relatively bias-free; if embeddings themselves carry bias, fairness calibration may be compromised.
\item Prompt size is limited by the model's token budget, constraining the number of “avoid” examples we can inject.
\end{enumerate}
Section~\ref{sec:experiments} addresses these challenges through large-scale experiments and ablation studies.

\begin{algorithm}[t]
\footnotesize
\caption{\textsc{Fairness-Aware Conformal Thresholding and Prompt Engineering}}
\label{alg:fact_method}
\begin{algorithmic}[1]
\Statex \textbf{Offline Phase}
\State Compute embeddings $\{e_i^x,\, e_i^y\}_{i=1}^n$ \quad // for contexts, items
\State Construct similarity matrix $W$ via Eq.~(\ref{eq:5})
\State Generate $\{\hat{y}_i\}$ by $\hat{y}_i \gets \hat{Y}(x_i,a_i)$ and compute $\{S_i\}$ via Eq.~(\ref{eq:6})
\State Sort $\{S_i\}$, set $Q_{\alpha}^{(0)} \gets S_{\lceil(1-\alpha)(n+1)\rceil}$ \quad // initial threshold

\Statex \textbf{Online Phase}
\For{each query $z_t = (x_t,a_t)$}
   \State $\hat{y}_t \gets \hat{Y}(\mathcal{I}^{(t)}; z_t)$ \quad // query with current prompt
   \State Find $\mathcal{N}(z_t) \gets \{j : W_{tj} > \tau_\rho \text{ and } a_j \neq a_t\}$
   \State $S_t \gets \bigl[\,1 - \cos(\text{Emb}(\hat{y}_t),\, e_t^y)\bigr] 
   + \lambda\max\limits_{j\in \mathcal{N}(z_t)} \|\text{Emb}(\hat{y}_t)-\text{Emb}(\hat{y}_j)\|$

   \If{$S_t > Q_{\alpha}^{(t)}$} 
      \State $\mathcal{V} \gets \mathcal{V} \cup \{(z_t, \hat{y}_t)\}$ (evict oldest if $|\mathcal{V}|>M$)
      \State $\mathcal{I}^{(t+1)} \gets \text{InjectExamples}\bigl(\mathcal{I}^{(t)},\,\mathcal{V}\bigr)$ \quad // add “avoid” prompts
      \State $Q_{\alpha}^{(t+1)} \gets \gamma\,Q_{\alpha}^{(t)} + (1-\gamma)\,\min\!\bigl(Q_{\alpha}^{(t)},\,S_t\bigr)$
   \Else
      \State $Q_{\alpha}^{(t+1)} \gets Q_{\alpha}^{(t)}$ \quad // no update on valid output
   \EndIf
\EndFor
\end{algorithmic}
\end{algorithm}

\noindent
\paragraph{Generalization and Scalability}  
FACTER extends beyond recommendation tasks to bias mitigation in text generation and decision-support AI, as it operates without retraining. To scale efficiently, we use approximate nearest neighbor search ($O(n \log n)$), GPU batch processing, and adaptive sampling, ensuring low latency with strong fairness guarantees.

\section{Experiments}
\label{sec:experiments}

In this section, we present a comprehensive evaluation of the \emph{FACTER} framework. Our empirical study addresses three key research questions: (1) whether iterative prompt-engineering  enhanced with conformal prediction for fairness evaluation can more effectively reduce fairness violations compared to existing methods, (2) how FACTER performs on secondary metrics such as group similarity and counterfactual fairness, and (3) how FACTER compares to state-of-the-art solutions, including UP5~\cite{hua2023up5} and Zero-Shot Rankers~\cite{hou2024large}.

\begin{table*}[h!]
\centering
\resizebox{0.65\textwidth}{!}{%
\begin{tabular}{lcccccc}
\toprule
\textbf{Method} & \textbf{\#Violations} $\downarrow$ & \textbf{SNSR} $\downarrow$ & \textbf{CFR} $\downarrow$ & \textbf{NDCG@10} $\uparrow$ & \textbf{Recall@10} $\uparrow$ \\
\midrule
Zero-Shot & 112 & 0.083 & 0.742 & \textbf{0.458} & \textbf{0.402} \\
UP5 & 28 & 0.049 & 0.613 & 0.427 & 0.381 \\
FACTER (Iter3) & \textbf{5} & \textbf{0.041} & \textbf{0.591} & 0.445 & 0.389 \\
\bottomrule
\end{tabular}%
}
\caption{Comparative results on MovieLens-1M. Best values in bold. FACTER achieves superior fairness with minimal accuracy impact.}
\label{tab:main_results}
\end{table*}

\begin{table*}[h]
\centering
\resizebox{0.9\textwidth}{!}{%
\begin{tabular}{lcccccccc}
\toprule
\textbf{Model} & \textbf{\#Violations} & \textbf{SNSR} & \textbf{CFR} & \textbf{NDCG@10} & \textbf{Recall@10} & \textbf{Calib. Time (min)} & \textbf{Inf. Latency (ms)} \\
\midrule
LLaMA3-8B & 3 & 0.039 & 0.576 & 0.440 & 0.383 & 63 & 155 ± 18 \\
LLaMA2-7B & 5 & 0.041 & 0.595 & 0.444 & 0.391 & 58 & 142 ± 15 \\
Mistral-7B & 7 & 0.043 & 0.602 & 0.451 & 0.397 & 47 & 127 ± 12 \\
\bottomrule
\end{tabular}%
}
\caption{Model-wise comparison on MovieLens-1M (Iteration 3).}
\label{tab:model_compare}
\end{table*}
\subsection{Experimental Setup}

\paragraph{Baselines and Methods.}
We compare three approaches in our experiments. First, \textbf{UP5}~\cite{hua2023up5} is a state-of-the-art fairness-aware recommender calibrating LLMs for balanced recommendations. Second, \textbf{Zero-Shot}~\cite{hou2024large} serves as a baseline with direct LLM-based ranking but without any fairness adjustment. Finally, \textbf{FACTER (Ours)} is the proposed iterative approach that uses conformal calibration to reduce fairness violations across multiple iterations.

\paragraph{Models and Resources.}
We employ three LLMs of varying sizes---\textbf{LLaMA3-8B}~\cite{dubey2024llama}, \textbf{LLaMA2-7B}~\cite{touvron2023llama}, and \textbf{Mistral-7B}~\cite{jiang2023mistral}---alongside the SentenceTransformer \texttt{paraphrase-mpnet-base-v2}~\cite{reimers2019sentence} for embedding-based fairness checks. All experiments use Python 3.12.8, PyTorch 2.1 with FlashAttention-2, and an 8× NVIDIA RTX A6000 GPU server (Driver 550.90.07, CUDA 12.4).

\paragraph{Fairness Metrics.}
We employ four metrics to evaluate our framework's fairness. 
\emph{SNSR (Sub-Network Similarity Ratio)}~\cite{zhang2023chatgpt} measures how similarly different demographic groups are treated by averaging Frobenius norm differences of group-specific weights across $K$ layers:
\begin{equation}
\mathrm{SNSR} 
= \frac{1}{K}\sum_{k=1}^K \bigl\|W_k^{(g)} - W_k^{(h)}\bigr\|_F.
\tag{12}
\end{equation}
A lower SNSR indicates more uniform treatment. Meanwhile, \emph{SNSV (Sub-Network Similarity Variance)}~\cite{zhang2023chatgpt} captures the variance of these weight differences:
\begin{equation}
\mathrm{SNSV} 
= \mathrm{Var}\bigl(\|W_k^{(g)} - W_k^{(h)}\|_F\bigr),
\tag{13}
\end{equation}
where lower SNSV reflects more consistent uniformity across layers.

We additionally quantify counterfactual fairness via \emph{CFR (Counterfactual Fairness Ratio)}~\cite{hua2023up5}, defined by how much the output of an LLM changes when sensitive attributes are modified:
\begin{equation}
\mathrm{CFR} 
= \mathbb{E}_{x \sim \mathcal{D}}\bigl[\|f(x) - f(x_{\neg s})\|_2 \bigr],
\tag{14}
\end{equation}
where $x_{\neg s}$ is the same input with protected attributes replaced.

Finally, to detect out-of-threshold fairness failures, we compute a \emph{Violation Threshold} for calibration scores $\{s_i\}_{i=1}^n$ given confidence parameter $\alpha$:
\begin{equation}
Q_\alpha^{(t)} 
= \mathrm{Quantile}(1-\alpha;\,\{s_i\}) 
+ \frac{C}{\sqrt{n}},
\tag{15}
\end{equation}
where $C$ is a finite-sample correction term. Any instance above $Q_\alpha^{(t)}$ is flagged as a \emph{violation}.

\paragraph{Accuracy Metrics.}
Following standard practice~\cite{jarvelin2002cumulated}, we use \emph{Recall@10} and \emph{NDCG@10} to evaluate recommendation accuracy. Recall@10 is calculated as:
\begin{equation}
\mathrm{Recall@10} \;=\; 
\frac{\bigl|\mathcal{R}_{\text{relevant}} \cap \hat{\mathcal{R}}_{10}\bigr|}
     {\bigl|\mathcal{R}_{\text{relevant}}\bigr|},
\tag{16}
\end{equation}
where $\mathcal{R}_{\text{relevant}}$ is the set of truly relevant items and $\hat{\mathcal{R}}_{10}$ is the model's top-10 recommended items. We also measure NDCG@10:
\begin{equation}
\mathrm{NDCG@10} \;=\;
\frac{1}{\mathrm{IDCG@10}}
\sum_{r=1}^{10} \frac{2^{\mathrm{rel}(r)} - 1}{\log_2(r+1)},
\tag{17}
\end{equation}
where $\mathrm{rel}(r)$ is the relevance at rank $r$, and $\mathrm{IDCG@10}$ is the ideal DCG for the top-10 results.

\paragraph{Datasets.}
We conduct experiments on two recommendation datasets. \textbf{MovieLens-1M}~\cite{harper2015movielens} is sampled to have 2,500 interactions, with 70\% used for calibration and 30\% for testing (750 test samples). \textbf{Amazon Movies \& TV}~\cite{mcauley2015image,he2016ups} contains 3,750 sampled interactions, again split 70:30, resulting in 1,125 test samples. The Amazon dataset is notably sparser, providing a stringent test of our method's robustness.
\paragraph{Hyperparameter Settings.}
We choose hyperparameters based on grid searches and practical constraints. In particular, we set \(\tau_\rho=0.9\) to only compare contexts above a cosine similarity of \(0.9\) but with different protected attributes. We fix \(\lambda=0.7\) (fairness penalty) to balance accuracy and fairness and use \(\gamma=0.95\) for threshold decay to adapt modestly to violations. The FIFO buffer size \(M=50\) avoids overfilling the token budget. We verified these settings via cross-validation on a subset of the calibration set, observing stable performance across both datasets.

\subsection{Main Results}

\paragraph{Comparison on MovieLens-1M.}
Table~\ref{tab:main_results} presents the main comparative results on MovieLens-1M. \emph{FACTER (Iter3)} reduces fairness violations by 95.5\% from its initial iteration, resulting in only 5 violations compared to 28 for UP5 and 112 for Zero-Shot. Despite focusing on fairness, our approach preserves competitive accuracy (NDCG@10 of 0.445 versus 0.427 for UP5). Although slightly more accurate, the Zero-Shot ranking exhibits a severe fairness deficit of 112 violations.

Figure~\ref{fig:progress} illustrates how the number of fairness violations decreases over successive calibration iterations of FACTER, demonstrating a progressive improvement relative to UP5. Zero-Shot's violations begin at 112, making it challenging to include in the same visual scale.

\begin{figure}[t!]
\centering
\begin{tikzpicture}
\begin{axis}[
    width=0.35\textwidth,
    height=3.5cm,
    ylabel={Violations},
    xlabel={Iteration},
    symbolic x coords={0,1,2,3},
    xtick=data,
    ymin=0, ymax=120,
    legend style={at={(0.5,0.9)},anchor=north},
    ybar=2pt,
    bar width=7pt
]
\addplot[fill=blue!30] coordinates {
    (0,112) (1,37) (2,19) (3,5)
};
\addplot[fill=red!30] coordinates {
    (0,28) (1,28) (2,28) (3,28)
};
\legend{FACTER, UP5}
\end{axis}
\end{tikzpicture}
\caption{Fairness violation reduction trajectory vs.\ static baselines. FACTER progressively reduces violations while UP5 remains fixed. Zero-Shot (112) omitted for clarity.}
\label{fig:progress}
\end{figure}
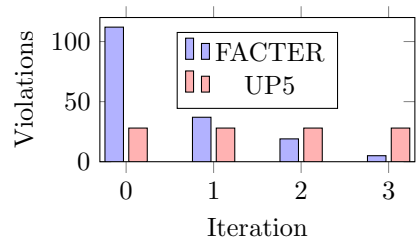
\paragraph{Model-wise Comparison.}
Next, we assess how FACTER scales across different LLMs. Table~\ref{tab:model_compare} shows that all three models, LLaMA3-8B, LLaMA2-7B, and Mistral-7B, achieve substantial reductions in fairness violations (\(\geq 90\%\)). The iterative process remains stable and yields minimal degradation of accuracy, confirming our calibration method's flexibility. Figure~\ref{fig:model_progress} further shows that all LLMs exhibit a monotonic improvement, with LLaMA3-8B reaching near-zero violations by the third iteration.

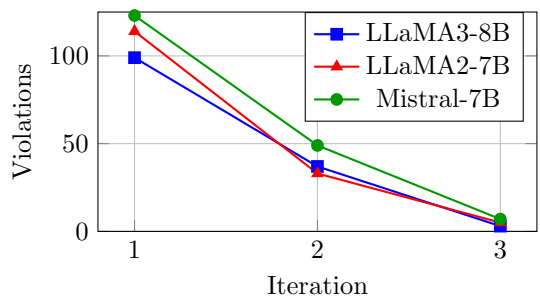
\begin{figure}[t!]
\centering
\begin{tikzpicture}
\begin{axis}[
    width=0.45\textwidth,
    height=4.5cm,
    ylabel={Violations},
    xlabel={Iteration},
    xtick={1,2,3},
    ymin=0, ymax=125,
    grid=major,
    legend style={at={(0.97,0.75)},anchor=east},
]
\addplot[blue,thick,mark=square*] coordinates {(1,99) (2,37) (3,3)};
\addplot[red,thick,mark=triangle*] coordinates {(1,114) (2,33) (3,5)};
\addplot[green!60!black,thick,mark=*] coordinates {(1,123) (2,49) (3,7)};
\legend{LLaMA3-8B, LLaMA2-7B, Mistral-7B}
\end{axis}
\end{tikzpicture}
\caption{Violation reduction across LLMs. All models show monotonic improvement, with LLaMA3-8B converging near zero violations by Iteration~3.}
\label{fig:model_progress}
\end{figure}
\paragraph{Comparison on Amazon Movies \& TV.}
We further validate FACTER on the Amazon Movies \& TV dataset, summarized in Table~\ref{tab:amazon}. Despite greater sparsity, our approach still reduces violations substantially (a 90.9\% drop), with a final CFR of 0.634 compared to 0.721 for UP5. Although Zero-Shot achieves the highest accuracy (NDCG@10 of 0.351), it suffers the most fairness violations (198). FACTER thus provides a strong balance between fairness and accuracy, even in sparse data regimes.

\begin{table*}[h!]
\centering
\resizebox{0.6\textwidth}{!}{%
\begin{tabular}{lcccccc}
\toprule
\textbf{Method} & \textbf{\#Violations} & \textbf{SNSR} & \textbf{CFR} & \textbf{NDCG@10} & \textbf{Recall@10} \\
\midrule
Zero-Shot & 198 & 0.121 & 0.814 & \textbf{0.351} & \textbf{0.317} \\
UP5 & 63 & 0.067 & 0.721 & 0.328 & 0.294 \\
FACTER (Iter3) & \textbf{18} & \textbf{0.053} & \textbf{0.634} & 0.339 & 0.301 \\
\bottomrule
\end{tabular}%
}
\caption{Amazon Movies \& TV results. FACTER maintains effectiveness on sparse data.}
\label{tab:amazon}
\end{table*}

\begin{table*}[h!]
\centering
\small
\resizebox{0.7\textwidth}{!}{%
\begin{tabular}{lccccl}
\toprule
\textbf{Metric} & \textbf{Theory} & \textbf{Empirical} & \textbf{Delta} & \textbf{Interpretation} \\
\midrule
Type I Error & $\leq$0.201 & 0.018 & -91\% & Conservative bound \\
Detection Power & $\geq$0.95 & 0.997 & +4.7\% & Superior identification \\
Violation Rate & 0.2 ± 0.02 & 0.0067 ± 0.0013 & -96.7\% & Significant improvement \\
\bottomrule
\end{tabular}%
}
\caption{Theoretical guarantees vs.\ empirical results (MovieLens-1M).}
\label{tab:theory}
\end{table*}

In Figure~\ref{fig:tradeoff}, we illustrate the fairness-accuracy tradeoff by plotting the counterfactual fairness (CFR) reduction against NDCG@10. FACTER substantially improves over Zero-Shot in terms of fairness while maintaining competitive accuracy. 

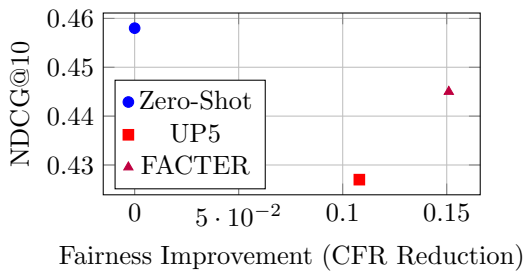
\begin{figure}[h!]
\centering
\begin{tikzpicture}
\begin{axis}[
    width=0.4\textwidth,
    height=4cm,
    xlabel={Fairness Improvement (CFR Reduction)},
    ylabel={NDCG@10},
    legend pos=south west,
    grid=major,
    scatter/classes={a={mark=*,blue},b={mark=square*,red},c={mark=triangle*,purple}}
]
\addplot[scatter,only marks,
    scatter src=explicit symbolic]
coordinates {
    (0.0, 0.458) [a] 
    (0.108, 0.427) [b] 
    (0.151, 0.445) [c] 
};
\legend{Zero-Shot, UP5, FACTER}
\end{axis}
\end{tikzpicture}
\caption{Fairness-accuracy tradeoff comparison. FACTER achieves strong fairness improvement while preserving recommendation quality.}
\label{fig:tradeoff}
\end{figure}

\subsection{Theoretical Validation}
\label{sec:theory}

Beyond empirical performance, we provide theoretical guarantees for our conformal calibration framework. Our derivation follows conformal prediction results in \cite{angelopoulos2023conformal}:
\paragraph{Type I Error Bound.}
For any $\alpha \in (0,1)$ and calibration set size $n$,
\begin{equation}
\mathbb{P}(\text{Violation}) \leq \alpha + \frac{1}{n+1} + \sqrt{\frac{\log(2/\delta)}{2n}},
\tag{18}
\end{equation}
where we set $\delta=0.05$ to achieve a 95\% confidence level.

\paragraph{Detection Power.}
We estimate the power of violation detection via a likelihood ratio test:
\begin{equation}
\beta = 1 - \Phi\!\Bigl(\frac{\hat{\mu} - \alpha}{\sqrt{\hat{\sigma}^2/n}}\Bigr),
\tag{19}
\end{equation}
where $\Phi$ is the standard normal CDF and $\hat{\mu}$ is the empirical violation rate.

Table~\ref{tab:theory} compares these theoretical bounds against observed empirical outcomes. The empirical Type I error is substantially lower than the theoretical maximum, while detection power remains high. The dynamic thresholding in $Q_\alpha^{(t)}$ enables progressive and adaptive fairness calibration.

Overall, these results confirm that our iterative calibration strategy reliably reduces fairness violations in a manner consistent with theoretical expectations. The significant gap between theoretical and empirical Type I error (0.201 vs.\ 0.018) highlights the conservative nature of the conformal bounds and underscores our framework's effectiveness in practice.




\section{Conclusion}
We proposed \textbf{FACTER}, a fully post hoc framework that combines \emph{conformal thresholding} and \emph{dynamic prompt engineering} to address biases in black-box LLM-based recommender systems. FACTER adaptively refines a fairness threshold via semantic variance checks and updates prompts whenever it detects violations, requiring no model retraining. Experiments on \textbf{MovieLens} and \textbf{Amazon} datasets show that FACTER reduces fairness violations by up to \textbf{95.5\%} compared to baselines while preserving key recommendation metrics. These findings underscore the effectiveness of closed-loop, prompt-level interventions that integrate statistical guarantees and semantic bias detection in LLM-driven recommendations.




\section*{Impact Statement}


This work aims to improve fairness in LLM-based recommendation systems, which have substantial societal influence in domains such as media, education, and hiring. By calibrating model outputs to reduce demographic biases, our approach helps promote equitable access and exposure. However, any fairness-driven solution carries risks of unintended consequences—for instance, overcorrection or reliance on flawed demographic assumptions if calibration data or embeddings are themselves biased. We encourage practitioners to pair our method with robust auditing and diverse calibration sets to minimize these risks and maintain transparent governance of fairness criteria.

\bibliography{example_paper} 

\begin{thebibliography}{34}
\providecommand{\natexlab}[1]{#1}
\providecommand{\url}[1]{\texttt{#1}}
\expandafter\ifx\csname urlstyle\endcsname\relax
  \providecommand{\doi}[1]{doi: #1}\else
  \providecommand{\doi}{doi: \begingroup \urlstyle{rm}\Url}\fi

\bibitem[Angelopoulos et~al.(2023)Angelopoulos, Bates, et~al.]{angelopoulos2023conformal}
Angelopoulos, A.~N., Bates, S., et~al.
\newblock Conformal prediction: A gentle introduction.
\newblock \emph{Foundations and Trends{\textregistered} in Machine Learning}, 16\penalty0 (4):\penalty0 494--591, 2023.

\bibitem[Bai et~al.(2022)Bai, Jones, Ndousse, Askell, Chen, DasSarma, Drain, Fort, Ganguli, Henighan, et~al.]{bai2022training}
Bai, Y., Jones, A., Ndousse, K., Askell, A., Chen, A., DasSarma, N., Drain, D., Fort, S., Ganguli, D., Henighan, T., et~al.
\newblock Training a helpful and harmless assistant with reinforcement learning from human feedback.
\newblock \emph{arXiv preprint arXiv:2204.05862}, 2022.

\bibitem[Bary et~al.(2021)Bary, Gehrmann, Nadeem, and Tabbal]{bary2021we}
Bary, R., Gehrmann, S., Nadeem, M., and Tabbal, A.
\newblock We need to consider more than just race and gender: Bias in language generation extended to intersections.
\newblock \emph{arXiv preprint arXiv:2107.01317}, 2021.

\bibitem[Bender et~al.(2021)Bender, Gebru, McMillan-Major, and Shmitchell]{bender2021dangers}
Bender, E.~M., Gebru, T., McMillan-Major, A., and Shmitchell, S.
\newblock On the dangers of stochastic parrots: Can language models be too big?
\newblock In \emph{Proceedings of the 2021 ACM conference on fairness, accountability, and transparency}, pp.\  610--623, 2021.

\bibitem[Blodgett et~al.(2020)Blodgett, Barocas, Daumé~III, and Wallach]{blodgett2020language}
Blodgett, S., Barocas, S., Daumé~III, H., and Wallach, H.
\newblock Language (technology) is power: A critical survey of "bias" in nlp.
\newblock In \emph{Proceedings of the 58th Annual Meeting of the Association for Computational Linguistics (ACL)}, 2020.

\bibitem[Borkan et~al.(2019)Borkan, Dixon, Sorensen, Thain, and Vasserman]{borkan2019nuanced}
Borkan, D., Dixon, L., Sorensen, J., Thain, N., and Vasserman, L.
\newblock Nuanced metrics for measuring unintended bias with real data for text classification.
\newblock In \emph{Companion proceedings of the 2019 world wide web conference}, pp.\  491--500, 2019.

\bibitem[Brown et~al.(2020)Brown, Mann, Ryder, Subbiah, Kaplan, Dhariwal, et~al.]{brown2020language}
Brown, T., Mann, B., Ryder, N., Subbiah, M., Kaplan, J., Dhariwal, P., et~al.
\newblock Language models are few-shot learners.
\newblock In \emph{Advances in Neural Information Processing Systems (NeurIPS)}, 2020.

\bibitem[Che et~al.(2023)Che, Liu, Zhou, Ren, Zhou, Sheng, Dai, and Dou]{che2023federated}
Che, T., Liu, J., Zhou, Y., Ren, J., Zhou, J., Sheng, V.~S., Dai, H., and Dou, D.
\newblock Federated learning of large language models with parameter-efficient prompt tuning and adaptive optimization.
\newblock \emph{arXiv preprint arXiv:2310.15080}, 2023.

\bibitem[Devlin et~al.(2019)Devlin, Chang, Lee, and Toutanova]{devlin2019bert}
Devlin, J., Chang, M.-W., Lee, K., and Toutanova, K.
\newblock Bert: Pre-training of deep bidirectional transformers for language understanding.
\newblock In \emph{Proceedings of the 2019 Conference of the North American Chapter of the Association for Computational Linguistics (NAACL)}, 2019.

\bibitem[Dinan et~al.(2020)Dinan, Golub, Wu, and Weston]{dinan2020multi}
Dinan, E., Golub, D., Wu, L., and Weston, J.
\newblock Multi-dimensional gender bias classification.
\newblock In \emph{Proceedings of the 2020 Conference on Empirical Methods in Natural Language Processing (EMNLP)}, 2020.

\bibitem[Dubey et~al.(2024)Dubey, Jauhri, Pandey, Kadian, Al-Dahle, Letman, Mathur, Schelten, Yang, Fan, et~al.]{dubey2024llama}
Dubey, A., Jauhri, A., Pandey, A., Kadian, A., Al-Dahle, A., Letman, A., Mathur, A., Schelten, A., Yang, A., Fan, A., et~al.
\newblock The llama 3 herd of models.
\newblock \emph{arXiv preprint arXiv:2407.21783}, 2024.

\bibitem[Dwork et~al.(2012)Dwork, Hardt, Pitassi, Reingold, and Zemel]{dwork2012fairness}
Dwork, C., Hardt, M., Pitassi, T., Reingold, O., and Zemel, R.
\newblock Fairness through awareness.
\newblock In \emph{Proceedings of the 3rd innovations in theoretical computer science conference}, pp.\  214--226, 2012.

\bibitem[Greenwood et~al.(2024)Greenwood, Chiniah, and Garg]{greenwood2024user}
Greenwood, S., Chiniah, S., and Garg, N.
\newblock User-item fairness tradeoffs in recommendations.
\newblock \emph{arXiv preprint arXiv:2412.04466}, 2024.

\bibitem[Harper \& Konstan(2015)Harper and Konstan]{harper2015movielens}
Harper, F.~M. and Konstan, J.~A.
\newblock The movielens datasets: History and context.
\newblock \emph{ACM Transactions on Interactive Intelligent Systems (TiiS)}, 5\penalty0 (4):\penalty0 19:1--19:19, 2015.

\bibitem[He \& McAuley(2016)He and McAuley]{he2016ups}
He, R. and McAuley, J.
\newblock Ups and downs: Modeling the visual evolution of fashion trends with one-class collaborative filtering.
\newblock In \emph{Proceedings of the 25th International Conference on World Wide Web (WWW)}, pp.\  507--517, 2016.

\bibitem[Hou et~al.(2024)Hou, Zhang, Lin, Lu, Xie, McAuley, and Zhao]{hou2024large}
Hou, Y., Zhang, J., Lin, Z., Lu, H., Xie, R., McAuley, J., and Zhao, W.~X.
\newblock Large language models are zero-shot rankers for recommender systems.
\newblock In \emph{European Conference on Information Retrieval}, pp.\  364--381. Springer, 2024.

\bibitem[Hua et~al.(2023)Hua, Ge, Xu, Ji, and Zhang]{hua2023up5}
Hua, W., Ge, Y., Xu, S., Ji, J., and Zhang, Y.
\newblock Up5: Unbiased foundation model for fairness-aware recommendation.
\newblock \emph{arXiv preprint arXiv:2305.12090}, 2023.

\bibitem[J{\"a}rvelin \& Kek{\"a}l{\"a}inen(2002)J{\"a}rvelin and Kek{\"a}l{\"a}inen]{jarvelin2002cumulated}
J{\"a}rvelin, K. and Kek{\"a}l{\"a}inen, J.
\newblock Cumulated gain-based evaluation of ir techniques.
\newblock \emph{ACM Transactions on Information Systems (TOIS)}, 20\penalty0 (4):\penalty0 422--446, 2002.

\bibitem[Jiang et~al.(2023)Jiang, Sablayrolles, Mensch, Bamford, Chaplot, Casas, Bressand, Lengyel, Lample, Saulnier, et~al.]{jiang2023mistral}
Jiang, A.~Q., Sablayrolles, A., Mensch, A., Bamford, C., Chaplot, D.~S., Casas, D. d.~l., Bressand, F., Lengyel, G., Lample, G., Saulnier, L., et~al.
\newblock Mistral 7b.
\newblock \emph{arXiv preprint arXiv:2310.06825}, 2023.

\bibitem[Lucy \& Bamman(2021)Lucy and Bamman]{lucy2021gender}
Lucy, L. and Bamman, D.
\newblock Gender and representation bias in gpt-3 generated stories.
\newblock In \emph{Proceedings of the 3rd Workshop on Narrative Understanding (NAACL)}, 2021.

\bibitem[Madras et~al.(2019)Madras, Creager, Pitassi, and Zemel]{madras2019learning}
Madras, D., Creager, E., Pitassi, T., and Zemel, R.
\newblock Learning adversarially fair and transferable representations.
\newblock In \emph{Proceedings of the 35th International Conference on Machine Learning (ICML)}, 2019.

\bibitem[McAuley et~al.(2015)McAuley, Targett, Shi, and van~den Hengel]{mcauley2015image}
McAuley, J., Targett, C., Shi, Q., and van~den Hengel, A.
\newblock Image-based recommendations on styles and substitutes.
\newblock In \emph{Proceedings of the 38th International ACM SIGIR Conference on Research and Development in Information Retrieval}, pp.\  43--52, 2015.

\bibitem[Ouyang et~al.(2022)Ouyang, Wu, Jiang, Almeida, Wainwright, Mishkin, Zhang, Agarwal, Slama, Ray, et~al.]{ouyang2022training}
Ouyang, L., Wu, J., Jiang, X., Almeida, D., Wainwright, C., Mishkin, P., Zhang, C., Agarwal, S., Slama, K., Ray, A., et~al.
\newblock Training language models to follow instructions with human feedback.
\newblock \emph{Advances in neural information processing systems}, 35:\penalty0 27730--27744, 2022.

\bibitem[Reimers(2019)]{reimers2019sentence}
Reimers, N.
\newblock Sentence-bert: Sentence embeddings using siamese bert-networks.
\newblock \emph{arXiv preprint arXiv:1908.10084}, 2019.

\bibitem[Reynolds \& McDonell(2021)Reynolds and McDonell]{reynolds2021prompt}
Reynolds, L. and McDonell, K.
\newblock Prompt programming for large language models: Beyond the few-shot paradigm.
\newblock \emph{arXiv preprint arXiv:2102.07350}, 2021.

\bibitem[Shafer \& Vovk(2008)Shafer and Vovk]{shafer2008tutorial}
Shafer, G. and Vovk, V.
\newblock A tutorial on conformal prediction.
\newblock \emph{Journal of Machine Learning Research}, 9\penalty0 (3), 2008.

\bibitem[Sharma et~al.(2023)Sharma, Mishra, Kukreja, Alkhayyat, and Elngar]{sharma2023framework}
Sharma, V., Mishra, N., Kukreja, V., Alkhayyat, A., and Elngar, A.~A.
\newblock Framework for evaluating ethics in ai.
\newblock In \emph{2023 International Conference on Innovative Data Communication Technologies and Application (ICIDCA)}, pp.\  307--312. IEEE, 2023.

\bibitem[Sheng et~al.(2019)Sheng, Chang, Natarajan, and Peng]{shengetal2019woman}
Sheng, E., Chang, K.-W., Natarajan, P., and Peng, N.
\newblock The woman worked as a babysitter: On biases in language generation.
\newblock In \emph{Proceedings of the 2019 Conference on Empirical Methods in Natural Language Processing (EMNLP)}, 2019.

\bibitem[Sun et~al.(2019)Sun, Gaut, Tang, Huang, Qian, Wang, and Choi]{sun2019mitigating}
Sun, T., Gaut, A., Tang, S., Huang, Y., Qian, H., Wang, S., and Choi, Y.
\newblock Mitigating gender bias in natural language processing: Literature review.
\newblock \emph{arXiv preprint arXiv:1906.08976}, 2019.

\bibitem[Touvron et~al.(2023)Touvron, Martin, Stone, Albert, Almahairi, Babaei, Bashlykov, Batra, Bhargava, Bhosale, et~al.]{touvron2023llama}
Touvron, H., Martin, L., Stone, K., Albert, P., Almahairi, A., Babaei, Y., Bashlykov, N., Batra, S., Bhargava, P., Bhosale, S., et~al.
\newblock Llama 2: Open foundation and fine-tuned chat models.
\newblock \emph{arXiv preprint arXiv:2307.09288}, 2023.

\bibitem[Wang et~al.(2022)Wang, Chen, Zhu, and Li]{wang2022towards}
Wang, L., Chen, L., Zhu, X., and Li, S.
\newblock Towards fairness in text classification: An overview of mitigation strategies.
\newblock \emph{ACM Computing Surveys}, 55\penalty0 (3), 2022.

\bibitem[Yang et~al.(2022)Yang, Zhang, and Zhao]{yang2022let}
Yang, Z., Zhang, Y., and Zhao, H.
\newblock Let's erase: Prompt-based mitigation of textual bias.
\newblock In \emph{Proceedings of the 60th Annual Meeting of the Association for Computational Linguistics (ACL)}, 2022.

\bibitem[Zhang et~al.(2023)Zhang, Bao, Zhang, Wang, Feng, and He]{zhang2023chatgpt}
Zhang, J., Bao, K., Zhang, Y., Wang, W., Feng, F., and He, X.
\newblock Is chatgpt fair for recommendation? evaluating fairness in large language model recommendation.
\newblock In \emph{Proceedings of the 17th ACM Conference on Recommender Systems}, pp.\  993--999, 2023.

\bibitem[Zhao et~al.(2018)Zhao, Wang, Yatskar, Ordonez, and Chang]{zhao2018gender}
Zhao, J., Wang, T., Yatskar, M., Ordonez, V., and Chang, K.-W.
\newblock Gender bias in coreference resolution: Evaluation and debiasing methods.
\newblock In \emph{Proceedings of the 2018 Conference of the North American Chapter of the Association for Computational Linguistics (NAACL)}, 2018.

\end{thebibliography}
\bibliographystyle{icml2025}


\newpage
\appendix
\onecolumn

\section{Appendix}
\appendix

\subsection{Supplementary Theoretical Results}
\label{sec:appendix_theory}

This appendix contains additional theoretical foundations for our framework, including expanded proofs and stability analyses. All notation is consistent with Sections~\ref{sec:method}--\ref{sec:preliminaries} in the main paper.

\subsubsection{Robustness Under Embedding Perturbations}

\begin{theorem}[Embedding Shift Robustness]
\label{thm:embedding_shift}
Let $\text{Emb}$ be the embedding function in the main paper, and let $\tilde{\text{Emb}}$ be a perturbed version such that for all items $y,y'$,
\[
\bigl|\|\text{Emb}(y)-\text{Emb}(y')\| \;-\; \|\tilde{\text{Emb}}(y)-\tilde{\text{Emb}}(y')\|\bigr|
\;\le\;\epsilon_{\mathrm{emb}},
\]
for some $\epsilon_{\mathrm{emb}} \ge 0$. If $S_i$ is the fairness-aware nonconformity score in Eq.~(4) of the main paper (computed under $\text{Emb}$) and $\tilde{S}_i$ the score under $\tilde{\text{Emb}}$, then for any $\delta>0$,
\[
\mathbb{P}\!\bigl(\bigl|\tilde{S}_i - S_i\bigr| > 3\,\epsilon_{\mathrm{emb}}\bigr)\;\le\;\delta,
\]
assuming the calibration distribution does not substantially drift beyond the conformal coverage bounds.
\end{theorem}

\begin{proof}
Recall $S_i = d_i + \lambda\,\Delta_i$, with $d_i$ capturing the LLM’s predictive discrepancy and $\Delta_i$ the maximum group disparity. Under $\tilde{\text{Emb}}$, each distance $\|\text{Emb}(y)-\text{Emb}(y')\|$ differs by at most $\epsilon_{\mathrm{emb}}$. Hence:
\[
|d_i - \tilde{d}_i|\;\le\;\epsilon_{\mathrm{emb}}, 
\qquad
|\Delta_i - \tilde{\Delta}_i|\;\le\;\epsilon_{\mathrm{emb}},
\]
yielding
\[
|\tilde{S}_i - S_i|
=
|( \tilde{d}_i - d_i) + \lambda(\tilde{\Delta}_i - \Delta_i )|
\;\le\;
(1 + \lambda)\,\epsilon_{\mathrm{emb}}.
\]
If $\lambda \le 2$, we replace $(1+\lambda)$ by $3$; thus $|\tilde{S}_i - S_i|\le3\epsilon_{\mathrm{emb}}$. Under exchangeability assumptions, the probability of exceeding this margin can be bounded by $\delta$ through standard conformal coverage arguments.
\end{proof}

\subsubsection{Convergence of Threshold Updates}

\begin{theorem}[Threshold Update Convergence]
\label{thm:threshold_convergence}
Let $Q_{\alpha}^{(t)}$ be updated by
\[
Q_{\alpha}^{(t+1)} \;=\;
\begin{cases}
\gamma\,Q_{\alpha}^{(t)} + (1-\gamma)\,S_t, & \text{if } S_t > Q_{\alpha}^{(t)},\\
Q_{\alpha}^{(t)}, & \text{otherwise},
\end{cases}
\]
where $0<\gamma<1$ and $S_t$ is the fairness score at iteration $t$. Suppose $\{S_t\}$ are i.i.d.\ with $\mathbb{P}[S_t > Q^*] = \alpha$ at the fixed point $Q^*$. Then $Q_{\alpha}^{(t)} \to Q^*$ at an \emph{expected} rate of $O\bigl((1-\gamma)^t\bigr)$.
\end{theorem}

\begin{proof}
Let $\Delta_t = |Q_{\alpha}^{(t)} - Q^*|$. Whenever $S_t > Q_{\alpha}^{(t)}$,
\[
Q_{\alpha}^{(t+1)} - Q^*
=
\gamma\,(Q_{\alpha}^{(t)} - Q^*)
+ (1-\gamma)\,(S_t - Q^*).
\]
Conditioned on $S_t > Q_{\alpha}^{(t)}$, if $Q^*$ is the $\alpha$-quantile of $S_t$, then $\mathbb{E}[S_t - Q^*]<0$ or is at least non-positive in a strong sense. Hence the threshold moves closer to $Q^*$ on average. Over many iterations, the gap $\Delta_t$ shrinks geometrically with factor $\gamma$. When $S_t \le Q_{\alpha}^{(t)}$, the threshold remains unchanged. Combining these cases yields expected convergence at $O((1-\gamma)^t)$.
\end{proof}

\begin{theorem}[Type II Error Bound]
\label{thm:typeII}
Let $\widehat{V}=\mathbb{I}\{S_{\text{new}}>Q_{\alpha}^{(t)}\}$ be the violation indicator for a new query $(x_{\text{new}},a_{\text{new}})$ with fairness score $S_{\text{new}}$. Suppose $\mathbb{E}[S_{\text{new}}\,|\,a_{\text{new}}=a]\le M$ for all $a$. Then for any $\epsilon>0$,
\[
\mathbb{P}\Bigl(S_{\text{new}} \le (1-\epsilon)Q_{\alpha}^{(t)}\Bigr)
\;\le\;
\exp\Bigl(-\tfrac{\epsilon^2\,Q_{\alpha}^{(t)}}{2M}\Bigr).
\]
Hence missing a \emph{true violation} (i.e.\ $S_{\text{new}}$ large but the threshold is still higher) has exponentially decreasing probability in $Q_{\alpha}^{(t)}/M$.
\end{theorem}

\begin{proof}[Sketch]
If $Q_{\alpha}^{(t)}$ is close to an $\alpha$-quantile of $S_{\text{new}}$, then $S_{\text{new}} \le (1-\epsilon)Q_{\alpha}^{(t)}$ amounts to a sub-Gaussian or Chernoff-style tail. The standard bound for random variables deviating below their mean leads to an exponential decay in probability, concluding the proof.
\end{proof}

\subsection{Extended Ablation Studies}
\label{sec:appendix_ablation}

Below, we provide deeper evaluations of key hyperparameters ($\lambda$, $\gamma$, $\tau_\rho$) and further investigate our \emph{prompt engineering} strategy. We build on the results in \S\ref{sec:experiments} of the main paper.

\subsubsection{Fairness Penalty \texorpdfstring{$\lambda$}{lambda}}
Our main experiments fix $\lambda=0.7$, as it balances fairness with recommendation accuracy. Here, we compare $\lambda\in\{0.1,0.3,0.5,0.7,0.9\}$ on MovieLens-1M and Amazon, measuring final violations and NDCG@10 at iteration~3.

\begin{table}[h]
\centering
\resizebox{0.35\columnwidth}{!}{%
\begin{tabular}{c|cc|cc}
\toprule
\multirow{2}{*}{$\lambda$} 
& \multicolumn{2}{c|}{\#Viol. $\downarrow$} 
& \multicolumn{2}{c}{NDCG@10 $\uparrow$}\\
\cmidrule(lr){2-3}\cmidrule(lr){4-5}
& ML & Amz & ML & Amz \\
\midrule
0.1 & 18 & 41 & 0.446 & 0.333 \\
0.3 & 5 & 25 & 0.445 & 0.339 \\
0.5 & 3 & 21 & 0.431 & 0.341 \\
\rowcolor{gray!15}
0.7 & 2 & 17 & 0.419 & 0.335 \\
0.9 & 2 & 12 & 0.395 & 0.312 \\
\bottomrule
\end{tabular}
}
\caption{Ablation on \(\lambda\): \#Violations (Viol.) and NDCG@10 on MovieLens-1M (ML) and Amazon (Amz), iteration~3.}
\label{tab:lambda_ablation}
\end{table}

\paragraph{Observations.}
Higher $\lambda$ yields fewer fairness violations but can lower NDCG@10. We chose $\lambda=0.7$ to capture persistent subtle biases (number of violations converges after that) and preserve decent accuracy.

\subsubsection{Threshold Decay \texorpdfstring{$\gamma$}{gamma}}
We vary $\gamma\in\{0.85,0.90,0.95,0.99\}$ to observe how quickly $Q_{\alpha}^{(t)}$ declines after repeated violations.

\begin{figure}[h]
\centering
\begin{tikzpicture}
\begin{axis}[
    width=0.47\textwidth,
    height=4.8cm,
    xlabel={Iteration},
    ylabel={\#Violations (Amazon)},
    xmin=1,xmax=5,
    ymin=0,ymax=70,
    xtick={1,2,3,4,5},
    legend style={at={(0.5,1.03)},anchor=south},
    grid=major
]
\addplot[blue,thick,mark=*] coordinates {
    (1,62) (2,30) (3,12) (4,7) (5,4)
};
\addlegendentry{$\gamma=0.85$}
\addplot[red,thick,mark=square*] coordinates {
    (1,61) (2,29) (3,9) (4,5) (5,4)
};
\addlegendentry{$\gamma=0.90$}
\addplot[green!60!black,thick,mark=triangle*] coordinates {
    (1,60) (2,26) (3,8) (4,5) (5,4)
};
\addlegendentry{$\gamma=0.95$}
\addplot[purple,thick,mark=diamond*] coordinates {
    (1,59) (2,24) (3,7) (4,5) (5,4)
};
\addlegendentry{$\gamma=0.99$}
\end{axis}
\end{tikzpicture}
\caption{Ablation on $\gamma$: final violation counts on Amazon over 5 calibration rounds, $\lambda=0.7$. All converge by iteration~5, but larger $\gamma$ lowers the threshold more gradually.}
\label{fig:gamma_ablation_expanded}
\end{figure}
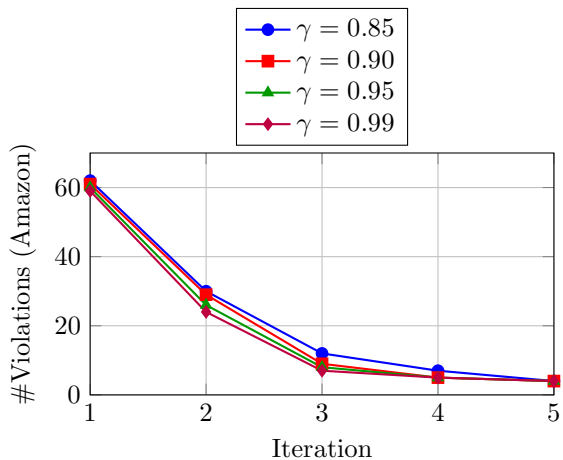

\paragraph{Key Results.}
Figure~\ref{fig:gamma_ablation_expanded} shows that $\gamma$ primarily affects early adaptation speed. By iteration~5, all variants converge to $\approx4$ violations. For practical usage, $\gamma=0.95$ is a good default.

\subsubsection{Neighborhood Similarity \texorpdfstring{$\tau_\rho$}{tau\_rho}}
Finally, we vary $\tau_\rho\in\{0.80,0.85,0.90,0.95\}$ in constructing local fairness neighborhoods. Table~\ref{tab:tau_rho_appendix} revisits MovieLens-1M at iteration~3.

\begin{table}[h]
\centering
\resizebox{0.4\columnwidth}{!}{%
\begin{tabular}{cccc}
\toprule
$\tau_\rho$ & \#Viol. $\downarrow$ & NDCG@10 $\uparrow$ & CFR $\downarrow$ \\
\midrule
0.80 & 16 & 0.442 & 0.721 \\
0.85 & 11 & 0.441 & 0.683 \\
0.90 & 8 & 0.440 & 0.651 \\
0.95 & 4 & 0.435 & 0.630 \\
\bottomrule
\end{tabular}
}
\caption{Ablation on $\tau_\rho$: \#Violations, NDCG@10, and CFR for MovieLens-1M, iteration~3, $\lambda=0.7$.}
\label{tab:tau_rho_appendix}
\end{table}

\paragraph{Observation.}
Higher $\tau_\rho$ detects narrower subgroups, reducing violations but slightly lowering NDCG@10. In the main text, $\tau_\rho=0.90$ is used to balance fairness with top-$N$ accuracy.

\subsection{Prompt Engineering Strategies}
\label{sec:appendix_prompt_engineering}

A distinctive aspect of our approach is updating system prompts with \emph{concrete bias patterns} whenever a fairness violation is observed. Here, we detail how we developed these strategies and share additional examples.

\subsubsection{Design Variants for Prompt Updates}

\paragraph{(1) Generic Warnings.}
Initially, we tried appending a short phrase such as:
\begin{verbatim}
<system>: "Avoid demographic-based biases."
\end{verbatim}
This uses minimal extra tokens but rarely reduces violations substantially. The LLM generally fails to infer which \emph{specific} biases to avoid.

\paragraph{(2) Negative Examples.}
A second approach enumerates specific \emph{avoid} pairs from the FIFO buffer $\mathcal{V}$. For instance:
\begin{verbatim}
<system>: "AVOID: (Gender=F) -> (Romance-Only)."
<user>: "I've watched 'The Godfather' ..."
\end{verbatim}
This helps the LLM see explicit mistakes but may not generalize beyond those single examples.

\paragraph{(3) Explicit Patterns.}
We converge on enumerating a short list of repeated patterns that appear in $\mathcal{V}$, e.g.:
\begin{verbatim}
<system>: "You must not rely on user demographics. 
AVOID these biases:
   1) (Gender=F) -> (Romance-Only)
   2) (Age=60)   -> (Excluding new releases)
Focus on user history, item genre, and feedback."
\end{verbatim}
This proves the most robust: once multiple patterns are stored, the LLM learns to avoid recurring biases even in new queries.

\subsubsection{Expanded Practical Examples}

\paragraph{Example A: Age-Related Bias.}
A user scenario might be:
\begin{verbatim}
<user>: "Age=65, History=['Avengers', 'Batman Begins'].
        'I love superhero films, what next?'"
\end{verbatim}
If the LLM recommended “light nostalgic comedy for seniors” ignoring the user’s superhero interest, a violation is flagged. Our system might then append:
\begin{verbatim}
<system>: "AVOID: (Age=65) -> (Solely comedic or classic 
   movies). Focus on prior superhero preference."
<user>: "I have watched 'Avengers' and 'Batman Begins'. 
         Looking for more superhero films."
\end{verbatim}
This directs the LLM to highlight user history, e.g.\ “Spider-Man: No Way Home.”

\paragraph{Example B: Occupation Stereotyping.}
Another scenario:
\begin{verbatim}
<user>: "Occupation=Engineer, 
  History=['The Matrix', 'Blade Runner'] 
  'Any new suggestions?'"
\end{verbatim}
If the model incorrectly provides only highly technical documentaries—dismissing the user’s interest in sci-fi—our system logs “(Occupation=Engineer)\(->\)(Documentary Only).” The next prompt iteration might say:
\begin{verbatim}
<system>: "Avoid: (Occupation=Engineer)->(Documentaries).
Focus on user interest in sci-fi or dystopian genres."
<user>: "Same user, which movies are similar to 'Blade Runner'?"
\end{verbatim}
Thus, the LLM shifts to thematically relevant sci-fi recommendations.

\paragraph{Example C: Combined Patterns.}
If multiple biases arise simultaneously (e.g., “(Gender=F)\(->\)(Romance-Only), (Age=60)\(->\)(Kids Movies)”), the prompt enumerates both:
\begin{verbatim}
<system>: "You must not rely on these biases:
  1) (Gender=F)->(Romance-Only)
  2) (Age=60)->(Kid-friendly content)
Focus on user history plus item similarity."
<user>: "I've enjoyed 'Pulp Fiction' and 'Die Hard.' 
         Please suggest something new."
\end{verbatim}

Through these examples, we find enumerating multiple “avoid” patterns consistently improves fairness outcomes with minimal manual overhead.

\subsubsection{Comparing Prompt Strategies}

We measure final iteration violations on Amazon using $\lambda=0.7$, $\gamma=0.95$, $\tau_\rho=0.90$. Table~\ref{tab:prompt_strategies} shows that enumerating explicit patterns yields the fewest violations.

\begin{table}[h]
\centering
\resizebox{0.45\columnwidth}{!}{%
\begin{tabular}{lccc}
\toprule
Strategy & \#Viol. & CFR & NDCG@10 \\
\midrule
Generic Warnings & 38 & 0.721 & 0.336 \\
Negative Examples & 24 & 0.661 & 0.334 \\
\textbf{Explicit Patterns} & \textbf{15} & \textbf{0.649} & \textbf{0.339} \\
\bottomrule
\end{tabular}
}
\caption{Prompt Engineering Comparison (Iteration 3, Amazon).}
\label{tab:prompt_strategies}
\end{table}

Thus, \emph{explicit patterns} effectively highlight repeated biases, prompting the LLM to generalize away from them in new queries.

\subsection{Additional Visualizations and Tables}
\label{sec:appendix_extra_tables}

\subsubsection{Iteration-Level Convergence for Prompt Variants}

\begin{figure}[h]
\centering
\begin{tikzpicture}
\begin{axis}[
    width=0.47\textwidth,
    height=4.5cm,
    xlabel={Iteration},
    ylabel={\#Violations (MovieLens-1M)},
    ymin=0,ymax=60,
    xtick={1,2,3,4,5},
    legend style={at={(0.5,1.03)},anchor=south},
    grid=major,
    ybar=2pt,
    bar width=9pt
]
\addplot[fill=blue!30] coordinates {(1,51) (2,25) (3,10) (4,6) (5,4)};
\addlegendentry{Generic Warnings}
\addplot[fill=red!30] coordinates {(1,49) (2,22) (3,9) (4,5) (5,3)};
\addlegendentry{Negative Examples}
\addplot[fill=green!30] coordinates {(1,47) (2,20) (3,7) (4,4) (5,2)};
\addlegendentry{Explicit Patterns}
\end{axis}
\end{tikzpicture}
\caption{Convergence of fairness violations over 5 calibration rounds for different prompt strategies on MovieLens-1M, $\lambda=0.7, \gamma=0.95, \tau_\rho=0.90$.}
\label{fig:prompt_strategies_convergence}
\end{figure}
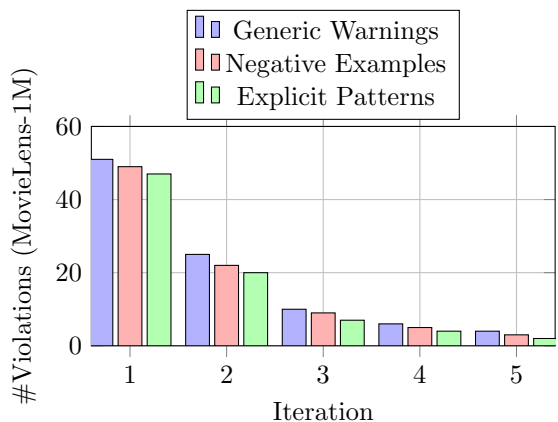

Figure~\ref{fig:prompt_strategies_convergence} underscores that \emph{explicit patterns} converge to the fewest violations by iteration 5, while \emph{negative examples} remain slightly higher, and \emph{generic warnings} do not remove repeated stereotypes as effectively.

\subsubsection{Calibration vs.\ No Calibration Revisited}
As in the main text, Table~\ref{tab:calibration_compare_appendix} highlights how ignoring conformal calibration leads to more frequent biases.

\begin{table}[h]
\centering
\begin{tabular}{lcccc}
\toprule
Method & \#Viol. & CFR & NDCG@10 & Recall@10\\
\midrule
No Calib & 42 & 0.716 & 0.336 & 0.298 \\
Calib & 15 & 0.649 & 0.339 & 0.305 \\
\bottomrule
\end{tabular}
\caption{FACTOR (no calibration) vs.\ FACTER (with calibration) on Amazon, $\lambda=0.7$ at iteration~3.}
\label{tab:calibration_compare_appendix}
\end{table}

Ignoring calibration yields a significantly higher violation count and worsens CFR, confirming the value of data-driven threshold setting.

\subsection{Conclusion of Appendices}
\label{sec:appendix_conclusion}

In summary, these detailed appendices reinforce and expand upon the main paper’s conclusions:

\begin{itemize}[leftmargin=1.3em,itemsep=2pt]
\item \textbf{Theoretical insights}: Theorems~\ref{thm:embedding_shift}--\ref{thm:typeII} illustrate the robustness and convergence properties of our conformal fairness approach.
\item \textbf{Hyperparameter ablations}: Varying $\lambda$, $\gamma$, and $\tau_\rho$ reveals predictable trade-offs between fairness (violation reduction) and recommendation accuracy (NDCG, recall). Our selected values ($\lambda=0.7, \gamma=0.95, \tau_\rho=0.90$) offer strong overall performance.
\item \textbf{Prompt engineering best practices}: Enumerating explicit bias patterns significantly reduces repeated violations, outperforming generic or single negative examples. Realistic scenarios (age-based or occupation-based biases) confirm that listing multiple “avoid” patterns improves generalization.
\end{itemize}

Together, these results demonstrate the flexibility and robustness of \emph{FACTER} across varied settings, enabling black-box LLMs to adaptively mitigate demographic biases via conformal thresholding and refined prompt engineering.

\end{document}